\documentclass[letterpaper, 10pt, conference]{ieeeconf}  

\IEEEoverridecommandlockouts                              

\usepackage{amsmath} 
\usepackage{amssymb}  

\usepackage{amsthm}
\usepackage{mathtools}
\usepackage{dsfont}
\usepackage{bm}
\usepackage{tikz}
\usepackage{pgfplots}
\usepackage{pgfplotstable}
\pgfplotsset{compat=1.13}
\usepackage[colorlinks=true, linkcolor=blue]{hyperref}
\usepackage{cite}
\usepackage{cleveref}

\theoremstyle{plain}
\newtheorem{theorem}{Theorem}
\newtheorem{corollary}[theorem]{Corollary}
\newtheorem{lemma}[theorem]{Lemma}
\newtheorem{proposition}[theorem]{Proposition}
\theoremstyle{definition}
\newtheorem{definition}[theorem]{Definition}
\theoremstyle{remark}
\newtheorem{remark}[theorem]{Remark}

\DeclareMathOperator{\diag}{diag}

\DeclareMathOperator{\im}{Im}
\DeclareMathOperator{\T}{\mathsf{T}}

\newcommand{\C}{\mathbb{C}}
\newcommand{\F}{\mathbb{F}}
\newcommand{\R}{\mathbb{R}}

\newcommand{\bE}{\bm{E}}
\newcommand{\bI}{\bm{I}}
\newcommand{\bV}{\bm{V}}

\newcommand{\cG}{\mathcal{G}}
\newcommand{\cI}{\mathcal{I}}
\newcommand{\cS}{\mathcal{S}}
\newcommand{\cX}{\mathcal{X}}

\newcommand{\hD}{\widehat{D}}
\newcommand{\hE}{\widehat{E}}
\newcommand{\hGamma}{\widehat{\Gamma}}
\newcommand{\hL}{\widehat{L}}
\newcommand{\hM}{\widehat{M}}
\newcommand{\hV}{\widehat{V}}
\newcommand{\hbE}{\widehat{\bE}}
\newcommand{\hbV}{\widehat{\bV}}
\newcommand{\hcG}{\widehat{\cG}}
\newcommand{\hdelta}{\widehat{\delta}}
\newcommand{\hf}{\widehat{f}}
\newcommand{\hn}{\widehat{n}}
\newcommand{\htheta}{\widehat{\theta}}

\newcommand{\ocG}{\overline{\cG}}
\newcommand{\on}{\overline{n}}

\newcommand{\imag}{\bm{\imath}}
\newcommand{\one}{\mathds{1}}
\newcommand{\LD}{L_{\mathrm{D}}}
\newcommand{\hLD}{\hL_{\mathrm{D}}}
\newcommand{\cl}{\mathrm{cl}}

\newcommand{\card}[1]{\lvert{#1}\rvert}

\newcommand{\myparen}[1]{\left(#1\right)}
\newcommand{\mybrack}[1]{\left[#1\right]}

\usetikzlibrary{calc}
\tikzset{
  generator/.style={
    circle,
    inner sep=0pt, minimum size=1.5em,
    very thick, draw
  },
  bus/.style={
    rectangle,
    inner sep=0pt, minimum width=0em, minimum height=1.5em,
    very thick, draw
  }
}

\title{\bfseries Synchronization and Aggregation of Nonlinear Power Systems with
  Consideration of Bus Network Structures}

\author{Petar~Mlinari\'{c}$^{1}$, Takayuki~Ishizaki$^{2}$,
  Aranya~Chakrabortty$^{3}$, Sara~Grundel$^{1}$, Peter~Benner$^{1}$, and
  Jun-ichi~Imura$^{2}$%
  \thanks{$^{1}$Computational Methods in Systems and Control Theory,
    Max Planck Institute for Dynamics of Complex Technical Systems,
    39106 Magdeburg, Germany: \newline
    \texttt{\string{mlinaric,grundel,benner\string}@mpi-magdeburg.mpg.de}}%
  \thanks{$^{2}$Department of Systems and Control Engineering,
    School of Engineering,
    Tokyo Institute of Technology,
    2-12-1, Meguro, Tokyo, Japan: \newline
    \texttt{\string{ishizaki,imura\string}@cs.e.titech.ac.jp}}%
  \thanks{$^{3}$Electrical \& Computer Engineering,
    North Carolina State University,
    Raleigh, NC 27695: \newline
    \texttt{aranya.chakrabortty@ncsu.edu}}%
}

\begin{document}
\maketitle
\thispagestyle{empty}
\pagestyle{empty}

\begin{abstract}
  We study nonlinear power systems consisting of generators, generator buses,
  and non-generator buses.
  First, looking at a generator and its bus' variables jointly, we introduce a
  synchronization concept for a pair of such joint generators and buses.
  We show that this concept is related to graph symmetry.
  Next, we extend, in two ways, the synchronization from a pair to a partition
  of all generators in the networks and show that they are related to either
  graph symmetry or equitable partitions.
  Finally, we show how an exact reduced model can be obtained by aggregating the
  generators and associated buses in the network when the original system is
  synchronized with respect to a partition, provided that the initial condition
  respects the partition.
  Additionally, the aggregation-based reduced model is again a power system.
\end{abstract}

\section{Introduction}
A power system is a network of electrical generators, loads, and their
associated control elements.
Each of these components may be thought of as nodes of a graph, while the
transmission lines connecting them can be regarded as the edges of the graph.
The nodes are modeled by physical laws that typically lead to a set of
differential equations.
These differential equations are coupled to each other across the edges.
One question that has been of interest to power engineers over many years is how
do the graph-theoretic properties of these types of electrical networks impact
system-theoretic properties of the grid model~\cite{AnnA13}.

In this work, we study synchronization properties of power systems
(see~\cite{DoeB14} for an overview) using graph-theoretic tools.
Specifically, we show relations to graph symmetry and equitable
partitions~\cite{RahJMetal09}, extending the work in~\cite{morIshKI16a} for
linear systems to nonlinear power systems.
Additionally, based on our results about synchronization, we propose a
structure-preserving, aggregation-based model order reduction framework for
nonlinear power systems.
Further, we show that for certain partitions this reduction is exact.
In general, the dynamics of the reduced system can be used to approximate the
dynamics of the original power system.

The motivation for model aggregation, in addition to reducing simulation time,
is the possibility to simulate or control only a certain part of the grid, or a
certain phenomenon that happens only over a certain time-scale.
Some recent work on aggregation of linear network systems can be found
in~\cite{morIshKIetal14, morIshKGetal15, morMliGB15, morCheKS16, morXueC16b,
morCheKS17}.

In \Cref{sec:system}, we describe the system we analyze.
Next, we introduce synchronization for a pair of generators and prove necessary
and sufficient conditions in \Cref{sec:sync-two}.
In \Cref{sec:partition-sync}, we continue in a similar way with two notions of
synchronization with respect to a partition.
We discuss aggregation-based reduction in \Cref{sec:aggregation}.
Finally, we give a demonstration of our results in \Cref{sec:example}.

\section{System Description}\label{sec:system}
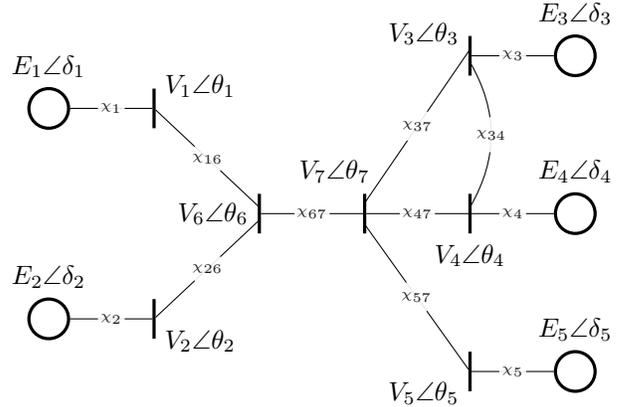
\begin{figure}[tb]
  \centering
  \begin{tikzpicture}[scale=1.4]

    \node[label=above:{$E_1 \angle \delta_1$}]
      (g1) [generator] at (-2, 1) {};
    \node[label=above:{$E_2 \angle \delta_2$}]
      (g2) [generator] at (-2, -1) {};
    \node[label=above:{$E_3 \angle \delta_3$}]
      (g3) [generator] at (3, 1.5) {};
    \node[label=above:{$E_4 \angle \delta_4$}]
      (g4) [generator] at (3, 0) {};
    \node[label=above:{$E_5 \angle \delta_5$}]
      (g5) [generator] at (3, -1.5) {};

    \node[label=above right:{$V_1 \angle \theta_1$}]
      (b1) [bus] at (-1, 1) {};
    \node[label=below right:{$V_2 \angle \theta_2$}]
      (b2) [bus] at (-1, -1) {};
    \node[label=above left:{$V_3 \angle \theta_3$}]
      (b3) [bus] at (2, 1.5) {};
    \node[label=below:{$V_4 \angle \theta_4$}]
      (b4) [bus] at (2, 0) {};
    \node[label=below left:{$V_5 \angle \theta_5$}]
      (b5) [bus] at (2, -1.5) {};
    \node[label=left:{$V_6 \angle \theta_6$}]
      (b6) [bus] at (0, 0) {};
    \node[label={[shift={(-0.4, 0)}]$V_7 \angle \theta_7$}]
      (b7) [bus] at (1, 0) {};

    \begin{scope}[every node/.style={
        circle,
        outer sep=0ex,
        inner sep=0.1ex,
        fill=white,
        fill opacity=0.9,
        text opacity=1}]

      \foreach \i in {1, ..., 5} {
        \draw (g\i) edge node {\tiny $\chi_{\i}$} (b\i);
      }

      \draw (b1) edge node {\tiny $\chi_{16}$} ($(b6)!0.3!(b6.north)$);
      \draw (b2) edge node {\tiny $\chi_{26}$} ($(b6)!0.3!(b6.south)$);
      \draw ($(b3)!0.3!(b3.south)$) to[bend left=30]
        node {\tiny $\chi_{34}$} ($(b4)!0.3!(b4.north)$);
      \draw ($(b3)!0.3!(b3.north)$) edge
        node {\tiny $\chi_{37}$} ($(b7)!0.5!(b7.north)$);
      \draw (b4) edge node {\tiny $\chi_{47}$} (b7);
      \draw (b5) edge node {\tiny $\chi_{57}$} ($(b7)!0.5!(b7.south)$);
      \draw (b6) edge node {\tiny $\chi_{67}$} (b7);
    \end{scope}
  \end{tikzpicture}
  \caption{Power system consisting of generators (circles) and buses (vertical
    bars), where the $i$th generator is only connected to the $i$th bus.
    See \Cref{tab:notation} for the notation.
  }\label{fig:power-system}
\end{figure}
\begin{table}[tb]
  \centering
  \caption{Notation}\label{tab:notation}
  \begin{tabular}{ll}
    Symbol & Description \\
    \hline
    $\imag$
           & imaginary unit ($\imag^2 = -1$) \\
    $\R$, $\C$
           & fields of real and complex numbers \\
    ${[a_{i}]}_{i \in S}$
           & vector $(a_{i_1}, a_{i_2}, \ldots, a_{i_n})$, if
             $S = \{i_1, i_2, \ldots, i_n\}$ \\
    $\diag(a)$
           & diagonal matrix with $a$ as its diagonal \\
    $A \circ B$
           & Hadamard (element-wise) product of two matrices \\
    $\one_n$
           & vector of ones of length $n$ \\
    $\one$
           & vector of ones with the length clear from context \\
    $e_i$
           & the $i$th column of the identity matrix \\
    $e_{S}$
           & matrix $[e_{i_1} \ e_{i_2} \ \cdots \ e_{i_n}]$, if
             $S = \{i_1, i_2, \ldots, i_n\}$ \\
    $\im{A}$
           & column space of matrix $A \in \F^{n \times m}$ \\
    $\card{S}$
           & cardinality of set $S$ \\
    $\sin$, $\cos$
           & functions applied element-wise to a vector or a matrix \\
    $\cG$
           & label set of generator buses \\
    $\ocG$
           & label set of non-generator buses \\
    $\bE_{\cG}(t)$
           & voltages of the generators at time $t$ \\
    $E_i$
           & voltage amplitude of the $i$th generator \\
    $\delta_i(t)$
           & voltage phase of the $i$th generator at time $t$ \\
    $\bV_{\cG}(t)$
           & voltages of the generator buses at time $t$ \\
    $\bV_{\ocG}(t)$
           & voltages of the non-generator buses at time $t$ \\
    $V_i(t)$
           & voltage amplitude of the $i$th bus at time $t$ \\
    $\theta_i(t)$
           & voltage phase of the $i$th bus at time $t$ \\
    $\bI_{\cG}(t)$
           & currents from generators to generator buses at time $t$ \\
    $\chi_{i}$
           & reactance between the $i$th generator and its bus \\
    $\chi_{ij}$
           & reactance between the $i$th and $j$th bus \\
    $\LD$
           & reactance matrix, $\diag({[\chi_{i}^{-1}]}_{i \in \cG})$ \\
    $L$
           & weighted graph Laplacian of the reactance network \\
    $\delta(t)$
           & ${[\delta_{i}(t)]}_{i \in \cG}$ \\
    $M$
           & diagonal matrix of inertias $M_i$ of the generators \\
    $D$
           & diagonal matrix of dissipativies $D_i$ of the generators \\
    $f$
           & vector of powers $f_i$ to the generators \\
    $X$
           & ${(\LD + L_{11} - L_{12} L_{22}^{-1} L_{12}^{\T})}^{-1} \LD$ \\
    $\Gamma$
           & $\LD {(\LD + L_{11} - L_{12} L_{22}^{-1} L_{12}^{\T})}^{-1} \LD$ \\
    $\gamma_{ij}$
           & ${[\Gamma]}_{ij}^{-1}$ \\
    $E$
           & ${[E_{i}]}_{i \in \cG}$ \\
    $V_{\cG}(t)$
           & ${[V_{i}(t)]}_{i \in \cG}$ \\
    $\theta_{\cG}(t)$
           & ${[\theta_{i}(t)]}_{i \in \cG}$ \\
    $\cX_{ij}$
           & subspace of synchronism $\{x \in \R^{n} : x_{i} = x_{j}\}$ \\
    $\cS_{ij}$
           & set of symmetrical matrices
             $\{A \in \R^{n \times n} : A \Pi_{ij} = \Pi_{ij} A\}$ \\
    $\cX_{\cl}$
           & $\bigcap_{\ell \in \hcG} \bigcap_{i, j \in \cI_{\ell}} \cX_{ij}$ \\
    $\cS_{\cl}$
           & $\bigcap_{\ell \in \hcG} \bigcap_{i, j \in \cI_{\ell}} \cS_{ij}$
  \end{tabular}
\end{table}
We use the power system example in \Cref{fig:power-system} to introduce the type
of system we analyze and to illustrate our results.
As in the example in \Cref{fig:power-system}, we consider power systems
consisting of generators and buses, where each generator is connected to exactly
one bus and buses can be classified into \emph{generator buses} (those connected
to one generator and some buses) and \emph{non-generator buses} (those connected
only to other buses).
We follow the classical model of a synchronous generator~\cite{Kun94}, which
means that the generators' voltage amplitude is constant over time $t$.

Let $\cG := \{1, 2, \ldots, n\}$ and $\ocG := \{n + 1, n + 2, \ldots, n + \on\}$
denote the label sets of generator and non-generator buses.
In the example in \Cref{fig:power-system}, we have $n = 5$ and $\on = 2$.
The vector of currents from generators to generator buses is given as
\begin{align}\label{eq:crnt}
  \bI_{\cG}(t)
  & =
    \frac{1}{\imag} \LD \myparen{\bE_{\cG}(t) - \bV_{\cG}(t)},
\end{align}
where the vectors of voltages of generators and generator buses are denoted as
\begin{align*}
  \bE_{\cG}(t)
  & :=
    \mybrack{E_{i}
      (\cos{\delta_{i}(t)} + \imag \sin{\delta_{i}(t)})}_{i \in \cG}
    \in \C^{n}, \\
  \bV_{\cG}(t)
  & :=
    \mybrack{V_{i}(t)
      (\cos{\theta_{i}(t)} + \imag \sin{\theta_{i}(t)})}_{i \in \cG}
    \in \C^{n},
\end{align*}
and $\LD$ is a positive diagonal reactance matrix given as
\begin{align*}
  \LD & := \diag\!\myparen{\mybrack{\chi_{i}^{-1}}_{i \in \cG}},
\end{align*}
where $\chi_{i}$ is the reactance between the $i$th generator and its bus
(see \Cref{fig:power-system}).
We assume the generator voltage amplitudes $E_i$ and reactances $\chi_i$ are
given constants.
Additionally, we assume the line resistances to be negligible.

The relation between the currents and voltages is given as
\begin{align}\label{eq:baleq}
  \begin{bmatrix}
    \bI_{\cG}(t) \\
    0
  \end{bmatrix}
  & =
    \frac{1}{\imag}
    \begin{bmatrix}
      L_{11} & L_{12} \\
      L_{12}^{\T} & L_{22}
    \end{bmatrix}
    \begin{bmatrix}
      \bV_{\cG}(t) \\
      \bV_{\ocG}(t)
    \end{bmatrix},
\end{align}
where the voltage vector of non-generator buses is denoted as
\begin{align*}
  \bV_{\ocG}(t)
  & :=
    \mybrack{V_{i}(t)
      (\cos{\theta_{i}(t)} + \imag \sin{\theta_{i}(t)})}_{i \in \ocG}
    \in \C^{\on}
\end{align*}
and $L = [L_{ij}] \in \R^{(n + \on) \times (n + \on)}$ denotes the weighted
graph Laplacian of the reactance network.
In particular, the $(i, j)$-th element of $L$ is $-\chi_{ij}^{-1}$ if the
$i$th and $j$th buses are connected (see \Cref{fig:power-system}) and the $i$th
diagonal element is $\sum_{j \neq i}{\chi_{ij}^{-1}}$.
In the following, we assume that the reactance network is connected, i.e.\ $L$
is irreducible.
This assumption can be made without loss of generality because the same
arguments can be applied to each connected component.
For the example in \Cref{fig:power-system} with $\chi_{ij} = 1$ for all $i, j$,
we have
\begin{align*}
  L
  & =
    \begin{bsmallmatrix*}[r]
      1 & 0 & 0 & 0 & 0 & -1 & 0 \\
      0 & 1 & 0 & 0 & 0 & -1 & 0 \\
      0 & 0 & 2 & -1 & 0 & 0 & -1 \\
      0 & 0 & -1 & 2 & 0 & 0 & -1 \\
      0 & 0 & 0 & 0 & 1 & 0 & -1 \\
      -1 & -1 & 0 & 0 & 0 & 3 & -1 \\
      0 & 0 & -1 & -1 & -1 & -1 & 4
    \end{bsmallmatrix*}.
\end{align*}

The dynamics of generators is given by
\begin{subequations}\label{eq:orig}
  \begin{align}\label{eq:orig-dyn}
    M \ddot{\delta}(t)
    + D \dot{\delta}(t)
    & =
      f - \mybrack{\frac{E_{i} V_{i}(t)}{\chi_{i}}
        \sin(\delta_{i}(t) - \theta_{i}(t))}_{i \in \cG},
  \end{align}
  with voltage phases $\delta(t) := {[\delta_{i}(t)]}_{i \in \cG}$, inertia
  constants $M := \diag\!\myparen{\mybrack{M_{i}}_{i \in \cG}}$, $M_i > 0$,
  damping constants $D := \diag\!\myparen{\mybrack{D_{i}}_{i \in \cG}}$,
  $D_i \ge 0$, and input powers $f \in \R^{n}$~\cite{Kun94}.
  Eliminating $\bI_{\cG}(t)$ from~\eqref{eq:crnt} and~\eqref{eq:baleq}, we
  obtain
  \begin{align}\label{eq:orig-cons}
    \begin{bmatrix}
      \LD \myparen{\bE_{\cG}(t) - \bV_{\cG}(t)} \\
      0
    \end{bmatrix}
    & =
      \begin{bmatrix}
        L_{11} & L_{12} \\
        L_{12}^{\T} & L_{22}
      \end{bmatrix}
      \begin{bmatrix}
        \bV_{\cG}(t) \\
        \bV_{\ocG}(t)
      \end{bmatrix},
  \end{align}
\end{subequations}
The set of equations~\eqref{eq:orig} forms a differential-algebraic system.
We can remove the algebraic constraints to find an equivalent set of
differential equations using Kron reduction~\cite{Kro39}.
First, from~\eqref{eq:orig-cons}, we find
\begin{align}
  \nonumber
  \bV_{\ocG}(t)
  & =
    -L_{22}^{-1} L_{12}^{\T} \bV_{\cG}(t), \\
  \label{eq:V=XE}
  \bV_{\cG}(t)
  & =
    X \bE_{\cG}(t),
\end{align}
where
\begin{align}\label{eq:X-def}
  X
  & :=
    \myparen{\LD + L_{11} - L_{12} L_{22}^{-1} L_{12}^{\T}}^{-1} \LD.
\end{align}
It follows that
\begin{align*}
  \Gamma
  & :=
    \LD \myparen{\LD + L_{11} - L_{12} L_{22}^{-1} L_{12}^{\T}}^{-1} \LD
    =
    \LD X
\end{align*}
is a positive definite matrix with positive elements, since
$\LD + L_{11} - L_{12} L_{22}^{-1} L_{12}^{\T}$ is positive definite and an
$M$-matrix (i.e., its eigenvalues have positive real parts and its off-diagonal
elements are nonpositive, which implies that the elements of its inverse are
positive).
We denote its elements by $\gamma_{ij}^{-1} := {[\Gamma]}_{ij}$.
Then, multiplying~\eqref{eq:V=XE} from the left by $\LD$, we find
\begin{align*}
  \mybrack{\frac{V_{i}(t)}{\chi_{i}} \cos{\theta_{i}(t)}}_{i \in \cG}
  & =
    \Gamma \mybrack{E_{i} \cos{\delta_{i}(t)}}_{i \in \cG}, \\
  \mybrack{\frac{V_{i}(t)}{\chi_{i}} \sin{\theta_{i}(t)}}_{i \in \cG}
  & =
    \Gamma \mybrack{E_{i} \sin{\delta_{i}(t)}}_{i \in \cG},
\end{align*}
which together with~\eqref{eq:orig-dyn} and the trigonometric identity
$\sin(\delta_{i}(t) - \theta_{i}(t))
= \sin{\delta_{i}(t)} \cos{\theta_{i}(t)}
- \cos{\delta_{i}(t)} \sin{\theta_{i}(t)}$ gives us
\begin{align*}
  \begin{aligned}
    &
      M \ddot{\delta}(t)
      + D \dot{\delta}(t) \\
    & =
      f - \bigl(
        \diag\!\myparen{\mybrack{E_{i} \sin{\delta_{i}(t)}}_{i \in \cG}}
        \Gamma
        \mybrack{E_{i} \cos{\delta_{i}(t)}}_{i \in \cG} \\
    & \qquad\qquad
        - \diag\!\myparen{\mybrack{E_{i} \cos{\delta_{i}(t)}}_{i \in \cG}}
        \Gamma
        \mybrack{E_{i} \sin{\delta_{i}(t)}}_{i \in \cG}
      \bigr).
  \end{aligned}
\end{align*}
Thus, now by using $\sin{\delta_{i}(t)} \cos{\delta_{j}(t)}
- \cos{\delta_{i}(t)} \sin{\delta_{j}(t)}
= \sin(\delta_{i}(t) - \delta_{j}(t))$, the Kron-reduced system
of~\eqref{eq:orig} is given as
\begin{subequations}
  \begin{align}
    \label{eq:kron-delta}
    M_i \ddot{\delta}_i(t)
    + D_i \dot{\delta}_i(t)
    & =
      f_i
      - \sum_{k = 1}^n{\frac{E_i E_k}{\gamma_{ik}}
        \sin(\delta_i(t) - \delta_k(t))},
    \intertext{with generator buses' voltages and phases satisfying}
    \label{eq:kron-V}
    \LD \bV_{\cG}(t) & = \Gamma \bE_{\cG}(t).
  \end{align}
\end{subequations}

Denoting $E := {[E_{i}]}_{i \in \cG}$, $V_{\cG}(t) := {[V_{i}(t)]}_{i \in \cG}$,
and $\theta_{\cG}(t) := {[\theta_{i}(t)]}_{i \in \cG}$, we can
write~\eqref{eq:orig-dyn} and~\eqref{eq:kron-delta} more compactly as
\begin{align}
  \label{eq:orig-rewrite}
  M \ddot{\delta}(t)
  + D \dot{\delta}(t)
  & =
    f - \LD \myparen{E \circ V_{\cG}(t)
      \circ \sin(\delta(t) - \theta_{\cG}(t))},
\end{align}
and
\begin{align*}
  \begin{aligned}
    &
      M \ddot{\delta}(t)
      + D \dot{\delta}(t) \\
    & =
      f
      - \myparen{\Gamma \circ E E^{\T}
        \circ \sin\!\myparen{\delta(t) \one_{n}^{\T}
          - \one_{n} {\delta(t)}^{\T}}} \one_{n}.
  \end{aligned}
\end{align*}

\section{Synchronization of Generator Pair}\label{sec:sync-two}
Let us denote the subspace of the synchronism between the $i$th and $j$th
elements by
\begin{align*}
  \cX_{ij} := \{x \in \R^{n} : x_{i} = x_{j}\}.
\end{align*}
In this notation, we introduce the following notion of synchronism for the
power system~\eqref{eq:orig}.
\begin{definition}
  Consider the power system~\eqref{eq:orig}.
  The $i$th and $j$th generators are said to be \textit{synchronized} if
  \begin{align*}
    \delta(t) \in \cX_{ij} \text{ and }
    \bV_{\cG}(t) \in \cX_{ij}, \text{ for all } t \geq 0
  \end{align*}
  and for any initial condition $\delta(0), \dot{\delta}(0) \in \cX_{ij}$.
\end{definition}
To characterize this generator synchronism in an algebraic manner, let us define
a set of symmetrical matrices with respect to the permutation of the $i$th and
$j$th columns and rows by
\begin{align}\label{eq:symmat}
  \cS_{ij}
  & :=
    \{A \in \R^{n \times n} : A \Pi_{ij} = \Pi_{ij} A\},
\end{align}
where $\Pi_{ij}$ denotes the permutation matrix associated with the $i$th and
$j$th elements, i.e., all diagonal elements of $\Pi_{ij}$ other than the $i$th
and $j$th elements are $1$, the $(i, j)$-th and $(j, i)$-th elements are $1$,
and the others are zero.
Note that $\cS_{ij}$ is not the set of usual symmetric (Hermitian) matrices; the
condition in~\eqref{eq:symmat} represents the invariance with respect to the
permutation of the $i$th and $j$th columns and rows, i.e.,
$\Pi_{ij}^{\T} A \Pi_{ij} = A$.
See \Cref{thm:symmetrical-equivalence} for equivalent conditions.

We state the main result about synchronization of a pair of generators and prove
it in the remainder of this Section.
\begin{theorem}\label{thm:Mi=Mj}
  Consider the power system~\eqref{eq:orig}.
  The following two statements hold.
  \begin{enumerate}
  \item\label{itm:Mi=Mj_n=2}
    Let $n = 2$ and $M_1 = M_2$.
    Then the two generators are synchronized if and only if $D_1 = D_2$,
    $f_1 = f_2$, and $E_1 = E_2$.
  \item\label{itm:Mi=Mj_n>=3}
    Let $n \ge 3$ and $M \in \cS_{ij}$.
    Then the $i$th and $j$th generators are synchronized if and only if
    $D \in \cS_{ij}$, $f \in \cX_{ij}$, $E \in \cX_{ij}$, and
    $\Gamma \in \cS_{ij}$.
  \end{enumerate}
\end{theorem}
\begin{remark}
  Essentially, this result shows that the $i$th and $j$th generators are
  synchronized when the system equation are invariant under swapping the $i$th
  and $j$th label.
\end{remark}
We arrange the proof of \Cref{thm:Mi=Mj} into a sequence of Propositions in this
Section, with some technical Lemmas in the Appendix.
We begin by analyzing the equations of the system~\eqref{eq:orig} without
assumptions on $n$ and $M$.
\begin{proposition}\label{thm:min}
  The $i$th and $j$th generators are synchronized if and only if
  \begin{align}
    \displaybreak[0]
    \label{eq:min1}
    \frac{D_i}{M_i}
    & =
      \frac{D_j}{M_j}, \\
    \displaybreak[0]
    \label{eq:min2}
    \frac{f_i}{M_i}
    & =
      \frac{f_j}{M_j}, \\
    \displaybreak[0]
    \label{eq:min3}
    \frac{E_i}{M_i \gamma_{ik}}
    & =
      \frac{E_j}{M_j \gamma_{jk}}, \text{ for } k \neq i, j, \\
    \displaybreak[0]
    \label{eq:min4}
    \frac{\chi_{i}}{\gamma_{ik}}
    & =
      \frac{\chi_{j}}{\gamma_{jk}}, \text{ for } k \neq i, j, \text{ and} \\
    \label{eq:min5}
    \frac{\chi_{i} E_i}{\gamma_{ii}}
    + \frac{\chi_{i} E_j}{\gamma_{ij}}
    & =
      \frac{\chi_{j} E_i}{\gamma_{ji}}
      + \frac{\chi_{j} E_j}{\gamma_{jj}}.
  \end{align}
\end{proposition}
\begin{proof}
  From~\eqref{eq:kron-delta}, we get
  \begin{multline*}
    \ddot{\delta}_i - \ddot{\delta}_j
    =
    -\frac{D_i}{M_i} \dot{\delta}_i
    + \frac{D_j}{M_j} \dot{\delta}_j
    + \frac{f_i}{M_i}
    - \frac{f_j}{M_j} \\
    - \sum_{k = 1}^n{\myparen{
      \frac{E_i E_k}{M_i \gamma_{ik}} \sin(\delta_i - \delta_k)
      - \frac{E_j E_k}{M_j \gamma_{jk}} \sin(\delta_j - \delta_k)}}.
  \end{multline*}
  It is clear that, if \eqref{eq:min1}, \eqref{eq:min2}, and \eqref{eq:min3} are
  true, then $\delta, \dot{\delta} \in \cX_{ij}$ implies
  $\ddot{\delta} \in \cX_{ij}$.
  For the other direction, let us assume that the $i$th and $j$th generators are
  synchronized.
  Then we necessarily have
  \begin{align*}
    &
      -\myparen{\frac{D_i}{M_i} - \frac{D_j}{M_j}} \dot{\delta}_i
      + \myparen{\frac{f_i}{M_i} - \frac{f_j}{M_j}} \\
    & \quad
      - \sum_{k = 1}^n{\myparen{\myparen{\frac{E_i E_k}{M_i \gamma_{ik}}
        - \frac{E_j E_k}{M_j \gamma_{jk}}} \sin(\delta_i - \delta_k)}}
      = 0,
  \end{align*}
  for any $\delta_i$, $\dot{\delta}_i$, and $\delta_k$, $k \neq i, j$.
  Choosing $\dot{\delta}_i = 0$ and $\delta_k = \delta_i$,
  condition~\eqref{eq:min2} follows.
  Taking $\dot{\delta}_i = 1$ and $\delta_k = \delta_i$, we find
  condition~\eqref{eq:min1}.
  Lastly, with $\delta_i - \delta_k = \frac{\pi}{2}$ for some $k \neq i, j$ and
  $\delta_i - \delta_{\ell} = 0$ for $\ell \neq i, j, k$,
  condition~\eqref{eq:min3} follows for the chosen $k$.

  From~\eqref{eq:kron-V}, we have
  \begin{align*}
    &
      V_i \cos{\theta_i}
      - V_j \cos{\theta_j} \\
    & =
      \myparen{\frac{\chi_{i} E_i}{\gamma_{ii}}
        - \frac{\chi_{j} E_i}{\gamma_{ji}}} \cos{\delta_i}
      + \myparen{\frac{\chi_{i} E_j}{\gamma_{ij}}
        - \frac{\chi_{j} E_j}{\gamma_{jj}}} \cos{\delta_j} \\
    \displaybreak[0]
    & \qquad
      + \sum_{\substack{k = 1 \\ k \neq i, j}}^n{
        \myparen{\frac{\chi_{i}}{\gamma_{ik}}
          - \frac{\chi_{j}}{\gamma_{jk}}} E_k \cos{\delta_k}}, \\
    &
      V_i \sin{\theta_i}
      - V_j \sin{\theta_j} \\
    & =
      \myparen{\frac{\chi_{i} E_i}{\gamma_{ii}}
        - \frac{\chi_{j} E_i}{\gamma_{ji}}} \sin{\delta_i}
      + \myparen{\frac{\chi_{i} E_j}{\gamma_{ij}}
        - \frac{\chi_{j} E_j}{\gamma_{jj}}} \sin{\delta_j} \\
    & \qquad
      + \sum_{\substack{k = 1 \\ k \neq i, j}}^n{
        \myparen{\frac{\chi_{i}}{\gamma_{ik}}
          - \frac{\chi_{j}}{\gamma_{jk}}} E_k \sin{\delta_k}}.
  \end{align*}
  Similarly, if we assume conditions~\eqref{eq:min4} and~\eqref{eq:min5} to be
  true, then $\delta_i = \delta_j$ implies
  $V_i \cos{\theta_i} = V_j \cos{\theta_j}$ and
  $V_i \sin{\theta_i} = V_j \sin{\theta_j}$, which in turn implies that
  $\bV_{\cG} \in \cX_{ij}$.
  Conversely, we have
  \begin{align*}
    0
    & =
      \myparen{\frac{\chi_{i} E_i}{\gamma_{ii}}
        + \frac{\chi_{i} E_j}{\gamma_{ij}}
        - \frac{\chi_{j} E_i}{\gamma_{ji}}
        - \frac{\chi_{j} E_j}{\gamma_{jj}}} \cos{\delta_i} \\
    \displaybreak[0]
    & \qquad
      + \sum_{\substack{k = 1 \\ k \neq i, j}}^n{
        \myparen{\frac{\chi_{i}}{\gamma_{ik}}
          - \frac{\chi_{j}}{\gamma_{jk}}} E_k \cos{\delta_k}}, \\
    0
    & =
      \myparen{\frac{\chi_{i} E_i}{\gamma_{ii}}
        + \frac{\chi_{i} E_j}{\gamma_{ij}}
        - \frac{\chi_{j} E_i}{\gamma_{ji}}
        - \frac{\chi_{j} E_j}{\gamma_{jj}}} \sin{\delta_i} \\
    & \qquad
      + \sum_{\substack{k = 1 \\ k \neq i, j}}^n{
        \myparen{\frac{\chi_{i}}{\gamma_{ik}}
          - \frac{\chi_{j}}{\gamma_{jk}}} E_k \sin{\delta_k}}.
  \end{align*}
  for arbitrary $\delta_i$ and $\delta_k$ for $k \neq i, j$.
  By appropriate choices of $\delta_i$ and $\delta_k$,
  conditions~\eqref{eq:min4} and~\eqref{eq:min5} follow.
\end{proof}
Let us now assume that $E_i \neq E_j$ and see what follows from conditions of
\Cref{thm:min}.
From~\eqref{eq:min4} and \Cref{thm:X1=1}, it follows that
$\frac{\chi_{i}}{\gamma_{ii}} + \frac{\chi_{i}}{\gamma_{ij}}
= \frac{\chi_{j}}{\gamma_{ji}} + \frac{\chi_{j}}{\gamma_{jj}}$.
Then, by~\eqref{eq:min5} and $E_i \neq E_j$, it is necessary that
$\frac{\chi_{i}}{\gamma_{ii}} = \frac{\chi_{j}}{\gamma_{ji}}$ and
$\frac{\chi_{i}}{\gamma_{ij}} = \frac{\chi_{j}}{\gamma_{jj}}$.
This, together with~\eqref{eq:min4}, means that the $i$th and $j$th rows in $X$
are equal, which is a contradiction with $X$ being invertible.
Therefore, for $i$th and $j$th generators to be synchronized, it is necessary
that $E_i = E_j$.
This allows us to simplify the statement of \Cref{thm:min}.
We can simplify it further by assuming $M_i = M_j$, which gives us the following
Corollary.
\begin{corollary}\label{thm:Mi=Mj_2}
  Let $M_i = M_j$.
  Then the $i$th and $j$th generators are synchronized if and only if
  \begin{align}
    \displaybreak[0]
    \nonumber
    D_i
    & =
      D_j, \\
    \displaybreak[0]
    \nonumber
    f_i
    & =
      f_j, \\
    \displaybreak[0]
    \nonumber
    E_i
    & =
      E_j, \\
    \displaybreak[0]
    \label{eq:MDf1}
    \gamma_{ik}
    & =
      \gamma_{jk}, \text{ for } k \neq i, j, \\
    \displaybreak[0]
    \label{eq:MDf2}
    \frac{\chi_{i}}{\gamma_{ik}}
    & =
      \frac{\chi_{j}}{\gamma_{jk}}, \text{ for } k \neq i, j, \text{ and} \\
    \label{eq:MDf3}
    \frac{\chi_{i}}{\gamma_{ii}}
    + \frac{\chi_{i}}{\gamma_{ij}}
    & =
      \frac{\chi_{j}}{\gamma_{ji}}
      + \frac{\chi_{j}}{\gamma_{jj}}.
  \end{align}
\end{corollary}
In the following, we separate the $n = 2$ and $n \ge 3$ cases.
First, we use \Cref{thm:Mi=Mj_2} to prove part~\ref{itm:Mi=Mj_n=2} of
\Cref{thm:Mi=Mj}.
\begin{proof}[{Proof of \Cref{thm:Mi=Mj}, part~\ref{itm:Mi=Mj_n=2}}]
  This is true since~\eqref{eq:MDf1} and~\eqref{eq:MDf2} are empty statements,
  while~\eqref{eq:MDf3} follows immediately from \Cref{thm:X1=1}.
\end{proof}
Finally, to prove part~\ref{itm:Mi=Mj_n>=3} of \Cref{thm:Mi=Mj}, we simplify the
statement of \Cref{thm:Mi=Mj_2} for the case of $n \ge 3$.
This gives us the following Corollary.
\begin{corollary}\label{thm:Mi=Mj_n>=3_2}
  Let $n \ge 3$ and $M_i = M_j$.
  Then the $i$th and $j$th generators are synchronized if and only if
  \begin{align}
    \displaybreak[0]
    \nonumber
    D_i & = D_j, \\
    \displaybreak[0]
    \nonumber
    f_i & = f_j, \\
    \displaybreak[0]
    \nonumber
    E_i & = E_j, \\
    \displaybreak[0]
    \label{eq:MDfE1}
    \gamma_{ik} & = \gamma_{jk}, \text{ for } k \neq i, j, \\
    \displaybreak[0]
    \label{eq:MDfE2}
    \chi_{i} & = \chi_{j}, \text{ and} \\
    \label{eq:MDfE3}
    \gamma_{ii} & = \gamma_{jj}.
  \end{align}
\end{corollary}
\begin{proof}
  From~\eqref{eq:MDf1} and~\eqref{eq:MDf2} follows~\eqref{eq:MDfE2}, using that
  there are at least three generators.
  Then, from~\eqref{eq:MDf3},~\eqref{eq:MDfE2}, and symmetry
  $\gamma_{ij} = \gamma_{ji}$ follows~\eqref{eq:MDfE3}.
\end{proof}
\Cref{thm:Mi=Mj_n>=3_2}, together with two Lemmas in the Appendix, allows us to
complete the proof of \Cref{thm:Mi=Mj}.
\begin{proof}[{Proof of \Cref{thm:Mi=Mj}, part~\ref{itm:Mi=Mj_n>=3}}]
  Conditions~\eqref{eq:MDfE1} and~\eqref{eq:MDfE3}, by
  \Cref{thm:symmetrical-equivalence}, are equivalent to $\Gamma \in \cS_{ij}$,
  which, by \Cref{thm:Gamma-symmetrical}, is in turn equivalent to
  $\LD \in \cS_{ij}$ and $L_{11} - L_{12} L_{22}^{-1} L_{12}^{\T} \in \cS_{ij}$.
  Therefore,~\eqref{eq:MDfE1} and~\eqref{eq:MDfE3} imply~\eqref{eq:MDfE2}.
\end{proof}

\section{Synchronization of Generator Partition}\label{sec:partition-sync}
Let $\cI = {\{\cI_{\ell}\}}_{\ell \in \hcG}$ be a partition of the set $\cG$,
where $\hcG = \{1, 2, \ldots, \hn\}$ and $\hn \le n$.
In particular, the clusters $\cI_{\ell}$ satisfy
\begin{enumerate}
\item $\cI_{\ell} \neq \emptyset$, for all $\ell \in \hcG$,
\item $\cI_{\ell_1} \cap \cI_{\ell_2} = \emptyset$, for all $\ell_1, \ell_2 \in
  \hcG$ such that $\ell_1 \neq \ell_2$, and
\item $\bigcup_{\ell \in \hcG}{\cI_{\ell}} = \cG$.
\end{enumerate}
Let us denote
\begin{align*}
  \cX_{\cl}
  & :=
    \bigcap_{\ell \in \hcG} \bigcap_{i, j \in \cI_{\ell}} \cX_{ij}, \quad
  \cS_{\cl}
    :=
    \bigcap_{\ell \in \hcG} \bigcap_{i, j \in \cI_{\ell}} \cS_{ij}.
\end{align*}
We define the aggregation matrix as
\begin{align*}
  P =
  \begin{bmatrix}
    e_{\cI_1} \one_{\card{\cI_1}}
    & e_{\cI_2} \one_{\card{\cI_2}}
    & \cdots
    & e_{\cI_{\hn}} \one_{\card{\cI_{\hn}}}
  \end{bmatrix}
  \in \R^{n \times \hn}.
\end{align*}
Notice that $\cX_{\cl} = \im{P}$.

We define two notions generalizing the synchronization of two generators to a
partition of generators.
\begin{definition}
  The system~\eqref{eq:orig} is said to be
  \emph{strongly synchronized with respect to partition $\cI$}
  if the $i$th and $j$th generators are synchronized for all
  $i, j \in \cI_{\ell}$ and all $\ell \in \hcG$, i.e.\ $\delta(t) \in \cX_{ij}$
  and $\bV_{\cG}(t) \in \cX_{ij}$ for all $t \ge 0$ and for any
  $\delta(0), \dot{\delta}(0) \in \cX_{ij}$, $i, j \in \cI_{\ell}$, and
  $\ell \in \hcG$.

  The system~\eqref{eq:orig} is said to be
  \emph{weakly synchronized with respect to partition $\cI$}
  if, for arbitrary $\delta(0), \dot{\delta}(0) \in \cX_{\cl}$, there exist
  functions $\hdelta : [0, \infty) \to \R^{\hn}$ and
  $\hbV_{\hcG} : [0, \infty) \to \C^{\hn}$ such that $\delta(t) = P \hdelta(t)$
  and $\bV_{\cG}(t) = P \hbV_{\hcG}(t)$, i.e.\ $\delta(t) \in \cX_{\cl}$ and
  $\bV_{\cG}(t) \in \cX_{\cl}$ for all $t \ge 0$ and for any
  $\delta(0), \dot{\delta}(0) \in \cX_{\cl}$.
\end{definition}
\begin{remark}
  Notice that strong synchronization is equivalent to
  $\cX_{ij} \times \cX_{ij} \times \cX_{ij}$ being an invariant set for
  $(\delta, \dot{\delta}, \hbV_{\hcG})$ for any $i, j \in \cI_{\ell}$ and
  $\ell \in \hcG$, while weak synchronization is equivalent to an invariant set
  being $\cX_{\cl} \times \cX_{\cl} \times \cX_{\cl}$.
  This means that, if the power system is strongly synchronized, when two
  generators and their buses in the same cluster have equal state, they will
  remain equal.
  If the power system is weakly synchronized, then when the states of every
  generator and its bus are equal to all others in the same cluster, they will
  stay equal.
  From this, we see that that if the system~\eqref{eq:orig} is strongly
  synchronized with respect to $\cI$, then it is also weakly synchronized with
  respect to $\cI$, since $\cX_{\cl} \times \cX_{\cl} \times \cX_{\cl} \subseteq
  \cX_{ij} \times \cX_{ij} \times \cX_{ij}$, for all $i, j \in \cI_{\ell}$ and
  all $\ell \in \hcG$.

  Further, the $i$th and $j$th generators are synchronized if and only
  if~\eqref{eq:orig} is either strongly or weakly synchronized with respect to
  $\{\{i, j\}\} \cup \{\{k\} : k \neq i, j\}$.

  Finally, notice that~\eqref{eq:orig} is always both strongly and weakly
  synchronized with respect to $\{\{i\} : i \in \cG\}$.
\end{remark}
In the following, we show necessary and sufficient conditions for the two
synchronization notions.
To start, in the next Proposition, we present cases when the structure of
$\Gamma$ has no influence.
It also illustrates the relation between strong and weak synchronization.
\begin{proposition}
  Let $\cI = \{\cG\}$, $M, D \in \cS_{\cl}$, and $f, E \in \cX_{\cl}$.
  Then the system~\eqref{eq:orig} is weakly synchronized with respect to
  $\{\cG\}$.
  If additionally $n = 2$, then~\eqref{eq:orig} is also strongly synchronized
  with respect to $\{\cG\}$.
\end{proposition}
\begin{proof}
  From the assumptions, it follows that $M = \widehat{m} I$,
  $D = \widehat{d} I$, $f = \hf \one$, and $E = \hE \one$, for some
  $\widehat{m} > 0$, $\widehat{d} \ge 0$, and $\hf, \hE \in \R$.
  Notice that for $\cI = \{\cG\}$, we have $P = \one$.

  Let us assume that $\delta(0), \dot{\delta}(0) \in \im{\one}$.
  To prove weak synchronization, we need to show that $\delta(t) \in \im{\one}$
  and $\bV_{\cG}(t) \in \im{\one}$.
  For the former, it is enough to show that $\ddot{\delta}(t) \in \im{\one}$ if
  $\delta(t), \dot{\delta}(t) \in \im{\one}$, which is clear, since then
  $\ddot{\delta}(t)
  = -M^{-1} D \dot{\delta}(t) + M^{-1} f
  = -\frac{\widehat{d}}{\widehat{m}} \dot{\delta}(t)
  + \frac{\hf}{\widehat{m}} \one$.
  For the latter, we see that
  $\bV_{\cG} = \LD^{-1} \Gamma \bE_{\cG} \in \im{\one}$ whenever
  $\bE_{\cG} \in \im{\one}$, which is equivalent to $\delta \in \im{\one}$.

  The second part follows from part~\ref{itm:Mi=Mj_n=2} of \Cref{thm:Mi=Mj}.
\end{proof}
We continue with the first main result of this Section---the necessary and
sufficient conditions for strong synchronization.
Here, symmetrical conditions for $\Gamma$ are relevant.
\begin{theorem}
  Let $n \ge 3$, $\cI$ arbitrary, and $M \in \cS_{\cl}$.
  Then the system~\eqref{eq:orig} is strongly synchronized with respect to $\cI$
  if and only if $D \in \cS_{\cl}$, $f \in \cX_{\cl}$, $E \in \cX_{\cl}$, and
  $\Gamma \in \cS_{\cl}$.
\end{theorem}
\begin{proof}
  It follows from applying part~\ref{itm:Mi=Mj_n>=3} of \Cref{thm:Mi=Mj} for
  every $i$th and $j$th generator where $i, j \in \cI_{\ell}$ and
  $\ell \in \hcG$.
\end{proof}
We conclude this Section with the second main result---the necessary and
sufficient conditions for weak synchronization.
Instead of symmetrical conditions, $\cX_{\cl}$ being $\Gamma$-invariant is one
of the conditions.
Since $\cX_{\cl} = \im{P}$, this actually means that $\cI$ is an equitable
partition for a graph whose adjacency matrix is $\Gamma$~\cite{GodR01}.
\begin{theorem}\label{thm:weak-sync}
  Let $\card{\cI} \ge 2$, $M, D \in \cS_{\cl}$, and $f, E \in \cX_{\cl}$.
  Then the system~\eqref{eq:orig} is weakly synchronized with respect to $\cI$
  if and only if
  \begin{align}
    \label{eq:weak-sync-cond}
    \LD \in \cS_{\cl} \quad
    \text{ and } \quad
    \cX_{\cl} \text{ is } \Gamma\text{-invariant}.
  \end{align}
\end{theorem}
\begin{proof}
  From the definition, we see that~\eqref{eq:orig} is weakly synchronized with
  respect to $\cI$ if and only if
  \begin{align}\label{eq:weak-sync-a}
    \begin{aligned}
      &
        \myparen{\forall \delta, \dot{\delta} \in \cX_{\cl}}\
        M^{-1} \Bigl(-D \dot{\delta} + f \\
      & \qquad
        - \myparen{\Gamma \circ E E^{\T}
          \circ \sin\!\myparen{\delta \one_{n}^{\T}
            - \one_{n} \delta^{\T}}} \one_{n}\Bigr)
        \in \cX_{\cl}
    \end{aligned}
  \end{align}
  and
  \begin{align}\label{eq:weak-sync-b}
    \myparen{\forall \delta \in \cX_{\cl}}\
    \LD^{-1} \Gamma \bE_{\cG} \in \cX_{\cl}.
  \end{align}
  Since $M, D \in \cS_{\cl}$ and $f \in \cX_{\cl}$,
  condition~\eqref{eq:weak-sync-a} is equivalent to
  \begin{align*}
    \myparen{\forall \delta \in \cX_{\cl}}\
    \myparen{\Gamma \circ E E^{\T}
      \circ \sin\!\myparen{\delta \one_{n}^{\T}
        - \one_{n} \delta^{\T}}} \one_{n}
    \in \cX_{\cl}.
  \end{align*}
  Using $\delta = P \hdelta$, $E = P \hE$, $\one_{n} = P \one_{\hn}$, and that
  $v \in \cX_{\cl}$ is equivalent to $\Pi_{ij} v = v$ for all
  $i, j \in \cI_{\ell}$ and $\ell \in \hcG$, we find that the above condition is
  equivalent to
  \begin{align}
    \nonumber
    &
      \myparen{\forall \hdelta \in \R^{\hn}}\!
      \myparen{\forall \ell \in \hcG}\!
      \myparen{\forall i, j \in \cI_{\ell}} \\
    \nonumber
    & \quad
      \myparen{(\Gamma P - \Pi_{ij} \Gamma P)
        \circ P \hE \hE^{\T}
        \circ P \sin\!\myparen{\hdelta \one_{\hn}^{\T}
          - \one_{\hn} \hdelta^{\T}}} \\
    \label{eq:weak-sync-1}
    & \qquad
      \times \one_{\hn} = 0.
  \end{align}
  In a similar way, we find the condition~\eqref{eq:weak-sync-b} is equivalent
  to
  \begin{align*}
    &
      \myparen{\forall \hbE_{\hcG} \in \C^{\hn}}\!
      \myparen{\forall \ell \in \hcG}\!
      \myparen{\forall i, j \in \cI_{\ell}} \\
    & \qquad\qquad
      \LD^{-1} \Gamma P \hbE_{\hcG}
      = \Pi_{ij} \LD^{-1} \Gamma P \hbE_{\hcG},
  \end{align*}
  or, more simply,
  \begin{align}
    \label{eq:weak-sync-2}
    \myparen{\forall \ell \in \hcG}\!
    \myparen{\forall i, j \in \cI_{\ell}}\
    \LD^{-1} \Gamma P
    = \Pi_{ij} \LD^{-1} \Gamma P.
  \end{align}

  It is straightforward to check that~\eqref{eq:weak-sync-cond}
  implies~\eqref{eq:weak-sync-1} and~\eqref{eq:weak-sync-2}.
  For the other direction, choosing $\hdelta = e_{\ell_2}$ for
  $\ell_2 \neq \ell$ in~\eqref{eq:weak-sync-1}, we find from the $i$th row that
  \begin{align}
    \label{eq:wscil1}
    \myparen{\forall \ell, \ell_2 \in \hcG, \ell_2 \neq \ell}\!
    \myparen{\forall i, j \in \cI_{\ell}}\
    \sum_{k \in \cI_{\ell_2}} \frac{1}{\gamma_{i k}}
    = \sum_{k \in \cI_{\ell_2}} \frac{1}{\gamma_{j k}}.
  \end{align}
  The $i$th row and $\ell_2$th column in condition~\eqref{eq:weak-sync-2} gives
  \begin{align}
    \label{eq:wscil2}
    \myparen{\forall \ell, \ell_2 \in \hcG}\!
    \myparen{\forall i, j \in \cI_{\ell}}\
    \chi_{i} \sum_{k \in \cI_{\ell_2}} \frac{1}{\gamma_{i k}}
    = \chi_{j} \sum_{k \in \cI_{\ell_2}} \frac{1}{\gamma_{j k}}.
  \end{align}
  Since the assumption is that there are at least two clusters in $\cI$,
  from~\eqref{eq:wscil1} and~\eqref{eq:wscil2} we find that
  $\chi_{i} = \chi_{j}$, for all $i, j \in \cI_{\ell}$ and all $\ell \in \hcG$,
  i.e., $\LD \in \cS_{\cl}$.
  This, together with~\eqref{eq:wscil2}, gives
  \begin{align*}
    \myparen{\forall \ell, \ell_2 \in \hcG}\!
    \myparen{\forall i, j \in \cI_{\ell}}\
    \sum_{k \in \cI_{\ell_2}} \frac{1}{\gamma_{i k}}
    = \sum_{k \in \cI_{\ell_2}} \frac{1}{\gamma_{j k}},
  \end{align*}
  which is equivalent to $\im(\Gamma P) \subseteq \cX_{\cl}$, i.e.\
  $\Gamma \cX_{\cl} \subseteq \cX_{\cl}$.
\end{proof}

\section{Aggregation of Power Systems}\label{sec:aggregation}
Let us assume that the system~\eqref{eq:orig} is weakly synchronized with
respect to a partition $\cI$.
Let also the initial condition satisfy
$\delta(0), \dot{\delta}(0) \in \cX_{\cl}$.
Then there exist $\hdelta$ and $\hbV_{\hcG}$ such that
$\delta(t) = P \hdelta(t)$ and $\bV_{\cG}(t) = P \hbV_{\hcG}(t)$, which also
gives us $V_{\cG}(t) = P \hV_{\hcG}(t)$ and
$\theta_{\cG}(t) = P \htheta_{\hcG}(t)$.
Inserting this into~\eqref{eq:orig} with dynamics rewritten as
in~\eqref{eq:orig-rewrite}, we find
\begin{gather*}
  \begin{aligned}
    &
      M P \ddot{\hdelta}(t)
      + D P \dot{\hdelta}(t) \\
    & =
      f
      - \LD \myparen{E \circ P \hV_{\hcG}(t)
        \circ \sin\!\myparen{P \hdelta(t) - P \htheta_{\hcG}(t)}},
  \end{aligned} \\
  \begin{bmatrix}
    \LD \myparen{\bE_{\cG}(t) - P \hbV_{\hcG}(t)} \\
    0
  \end{bmatrix}
  =
    \begin{bmatrix}
      L_{11} & L_{12} \\
      L_{12}^{\T} & L_{22}
    \end{bmatrix}
    \begin{bmatrix}
      P \hbV_{\hcG}(t) \\
      \bV_{\ocG}(t)
    \end{bmatrix}.
\end{gather*}
Assuming additionally that $E \in \cX_{\cl}$, i.e.\ $E = P \hE$ for some
$\hE \in \R^{\hn}$, and pre-multiplying the above dynamics and first block-row
of the constraint by $P^{\T}$, we obtain
\begin{subequations}\label{eq:orig-agg}
  \begin{gather}
    \begin{aligned}
      &
        \hM \ddot{\hdelta}(t)
        + \hD \dot{\hdelta}(t) \\
      & =
        \hf - \hLD \myparen{\hE \circ \hV_{\hcG}(t)
          \circ \sin\!\myparen{\hdelta(t) - \htheta_{\hcG}(t)}},
    \end{aligned} \\
    \begin{bmatrix}
      \hLD \myparen{\hbE_{\hcG}(t) - \hbV_{\hcG}(t)} \\
      0
    \end{bmatrix}
    =
      \begin{bmatrix}
        \hL_{11} & \hL_{12} \\
        \hL_{12}^{\T} & L_{22}
      \end{bmatrix}
      \begin{bmatrix}
        \hbV_{\hcG}(t) \\
        \bV_{\ocG}(t)
      \end{bmatrix},
  \end{gather}
\end{subequations}
where $\hM = P^{\T} M P$, $\hD = P^{\T} D P$, $\hf = P^{\T} f$,
$\hLD = P^{\T} \LD P$, $\hL_{11} = P^{\T} L_{11} P$, $\hL_{12} = P^{\T} L_{12}$.
Moreover, from $\delta(t) = P \hdelta(t)$ follows that
$\hdelta(0) = {(P^{\T} P)}^{-1} P^{\T} \delta(0)$ and
$\dot{\hdelta}(0) = {(P^{\T} P)}^{-1} P^{\T} \dot{\delta}(0)$.

Notice that the reduced model~\eqref{eq:orig-agg} is again a power system of the
same form as~\eqref{eq:orig}.
In particular, we have that $\hM$, $\hD$, and $\hLD$ are positive definite
diagonal matrices and that $\hL$ is a Laplacian matrix.
Additionally, note that this projection-based reduction can be done for
arbitrary power system and arbitrary partition.
In general, we can take~\eqref{eq:orig-agg} with
$\hdelta(0) = {(P^{\T} P)}^{-1} P^{\T} \delta(0)$,
$\dot{\hdelta}(0) = {(P^{\T} P)}^{-1} P^{\T} \dot{\delta}(0)$, and
$\hE = {(P^{\T} P)}^{-1} P^{\T} E$.
We can also apply Kron reduction to this reduced model.

\section{Illustrative Example}\label{sec:example}
For the example in \Cref{fig:power-system}, let $\chi_i = 1$ and $\chi_{ij} = 1$
for all $i, j$.
Then we have
\begin{align*}
  \Gamma
  & =
    \frac{1}{32}
    \begin{bsmallmatrix*}[r]
      21 & 5 & 2 & 2 & 2 \\
      5 & 21 & 2 & 2 & 2 \\
      2 & 2 & 16 & 8 & 4 \\
      2 & 2 & 8 & 16 & 4 \\
      2 & 2 & 4 & 4 & 20
    \end{bsmallmatrix*}.
\end{align*}
Additionally, let $M = D = I_5$, $f = 0$, and $E = \one_{5}$.
Then, using \Cref{thm:Mi=Mj}, we see that the first and second generators are
synchronized, and that the same is true for the third and fourth.
By definition, this implies that the system is strongly synchronized with
respect to $\{\{1, 2\}, \{3, 4\}, \{5\}\}$.
On the other hand, from \Cref{thm:weak-sync} and
\begin{align*}
  \Gamma
  \begin{bsmallmatrix*}[r]
    1 & 0 \\
    1 & 0 \\
    0 & 1 \\
    0 & 1 \\
    0 & 1
  \end{bsmallmatrix*}
  & =
    \frac{1}{16}
    \begin{bsmallmatrix*}[r]
      13 & 3 \\
      13 & 3 \\
      2 & 14 \\
      2 & 14 \\
      2 & 14
    \end{bsmallmatrix*}
    =
    \begin{bsmallmatrix*}[r]
      1 & 0 \\
      1 & 0 \\
      0 & 1 \\
      0 & 1 \\
      0 & 1
    \end{bsmallmatrix*}
    \myparen{%
    \frac{1}{16}
    \begin{bsmallmatrix*}[r]
      13 & 3 \\
      2 & 14
    \end{bsmallmatrix*}},
\end{align*}
we see that the  system is weakly synchronized with respect to
$\{\{1, 2\}, \{3, 4, 5\}\}$, but not strongly.
Using the partition $\cI = \{\{1, 2\}, \{3, 4, 5\}\}$ for aggregation, we find
the following reduced quantities:
$\hM = \hD =
\begin{bsmallmatrix}
  2 & 0 \\
  0 & 3
\end{bsmallmatrix}$,
$\hf = 0$, $\hE = \one_{2}$,
$\hLD = \hL_{11} =
\begin{bsmallmatrix}
  2 & 0 \\
  0 & 3
\end{bsmallmatrix}$,
$\hL_{12} =
\begin{bsmallmatrix}
  -2 & 0 \\
  0 & -3
\end{bsmallmatrix}$,
$\hGamma =
\frac{1}{8}
\begin{bsmallmatrix}
  13 & 3 \\
  3 & 21
\end{bsmallmatrix}$.
The \Cref{fig:partition} shows the partition and \Cref{fig:aggregated} the
associated reduced power system.
From the definition of weak synchronization, we know that this reduced power
system exactly reproduces the initial value response of the original system for
any initial condition $\delta(0), \dot{\delta}(0) \in \cX_{\cl}$, taking the
initial condition of the reduced model to be
$\hdelta(0) = {(P^{\T} P)}^{-1} P^{\T} \delta(0)$ and
$\dot{\hdelta}(0) = {(P^{\T} P)}^{-1} P^{\T} \dot{\delta}(0)$.
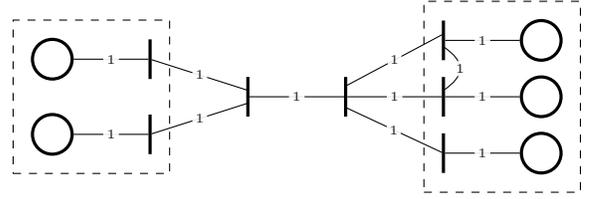
\begin{figure}[tb]
  \centering
  \begin{tikzpicture}[xscale=1.3, yscale=0.5]
    \node (g1) [generator] at (-2, 1) {};
    \node (g2) [generator] at (-2, -1) {};
    \node (g3) [generator] at (3, 1.5) {};
    \node (g4) [generator] at (3, 0) {};
    \node (g5) [generator] at (3, -1.5) {};

    \node (b1) [bus] at (-1, 1) {};
    \node (b2) [bus] at (-1, -1) {};
    \node (b3) [bus] at (2, 1.5) {};
    \node (b4) [bus] at (2, 0) {};
    \node (b5) [bus] at (2, -1.5) {};
    \node (b6) [bus] at (0, 0) {};
    \node (b7) [bus] at (1, 0) {};

    \begin{scope}[every node/.style={
        circle,
        outer sep=0ex,
        inner sep=0.1ex,
        fill=white,
        fill opacity=0.9,
        text opacity=1}]

      \foreach \i in {1, ..., 5} {
        \draw (g\i) edge node {\tiny $1$} (b\i);
      }

      \draw (b1) edge node {\tiny $1$} ($(b6)!0.3!(b6.north)$);
      \draw (b2) edge node {\tiny $1$} ($(b6)!0.3!(b6.south)$);
      \draw ($(b3)!0.3!(b3.south)$) to[bend left=30]
        node {\tiny $1$} ($(b4)!0.3!(b4.north)$);
      \draw ($(b3)!0.3!(b3.north)$) edge
        node {\tiny $1$} ($(b7)!0.5!(b7.north)$);
      \draw (b4) edge node {\tiny $1$} (b7);
      \draw (b5) edge node {\tiny $1$} ($(b7)!0.5!(b7.south)$);
      \draw (b6) edge node {\tiny $1$} (b7);
    \end{scope}

    \draw[dashed] ($(g1) + (-0.4, 1.04)$)
      -- ($(b1) + (0.2, 1.04)$)
      -- ($(b2) + (0.2, -1.04)$)
      -- ($(g2) + (-0.4, -1.04)$)
      -- cycle;
    \draw[dashed] ($(g3) + (0.4, 1.04)$)
      -- ($(b3) + (-0.2, 1.04)$)
      -- ($(b5) + (-0.2, -1.04)$)
      -- ($(g5) + (0.4, -1.04)$)
      -- cycle;
  \end{tikzpicture}
  \caption{Partition $\{\{1, 2\}, \{3, 4, 5\}\}$ applied to the original power
    system in \Cref{fig:power-system} with $\chi_i = \chi_{ij} = 1$ for all
    $i, j$.
  }\label{fig:partition}
\end{figure}
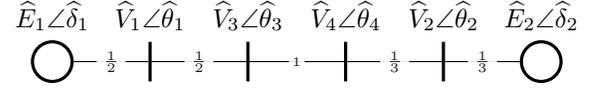
\begin{figure}[tb]
  \centering
  \begin{tikzpicture}[scale=1.3]
    \node[label=above:{$\hE_1 \angle \hdelta_1$}]
      (gr1) [generator] at (-2, 0) {};
    \node[label=above:{$\hE_2 \angle \hdelta_2$}]
      (gr2) [generator] at (3, 0) {};

    \node[label=above:{$\hV_1 \angle \htheta_1$}]
      (br1) [bus] at (-1, 0) {};
    \node[label=above:{$\hV_2 \angle \htheta_2$}]
      (br2) [bus] at (2, 0) {};
    \node[label=above:{$\hV_3 \angle \htheta_3$}]
      (br3) [bus] at (0, 0) {};
    \node[label=above:{$\hV_4 \angle \htheta_4$}]
      (br4) [bus] at (1, 0) {};

    \begin{scope}[every node/.style={
        circle,
        outer sep=0ex,
        inner sep=0.1ex,
        fill=white,
        fill opacity=0.9,
        text opacity=1}]

      \draw (gr1) edge node {\tiny $\frac{1}{2}$} (br1);
      \draw (gr2) edge node {\tiny $\frac{1}{3}$} (br2);

      \draw (br1) edge node {\tiny $\frac{1}{2}$} (br3);
      \draw (br2) edge node {\tiny $\frac{1}{3}$} (br4);
      \draw (br3) edge node {\tiny $1$} (br4);
    \end{scope}
  \end{tikzpicture}
  \caption{Reduced power system obtained by aggregating the system in
    \Cref{fig:partition} with $M = D = I_5$, $f = 0$, and $E = \one_{5}$.
  }\label{fig:aggregated}
\end{figure}

To demonstrate the possibility to aggregate using any partition, including those
with respect to which the power system is not weakly synchronized, and any
initial condition, we show simulation result for partition
$\{\{1, 2, 3\}, \{4, 5\}\}$ in \Cref{fig:simulation}.
We see that, in this case, the reduced model matches the steady state and
approximates the transient behavior.
Finding sufficient conditions for matching the steady state and deriving error
bounds is a possible topic of future research.
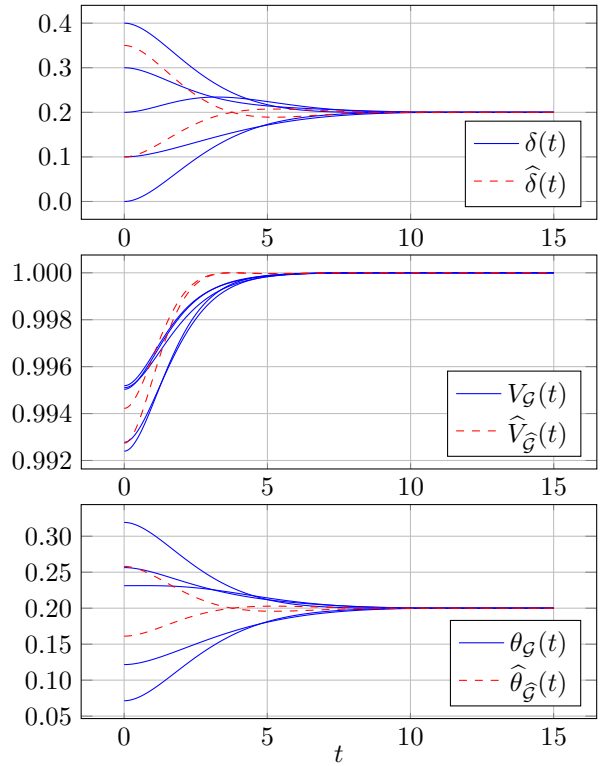
\begin{figure}[tb]
  \centering
  \pgfplotstableset{
    create on use/t/.style={
      create col/copy column from table={data/t.txt}{0}
    }
  }
  \begin{tikzpicture}
    \begin{axis}[
        name=delta,
        yscale=0.5,
        y tick label style={
          /pgf/number format/.cd,
          fixed,
          fixed zerofill,
          precision=1,
          /tikz/.cd
        },
        legend pos=south east,
        legend entries={$\delta(t)$,,,,,$\hdelta(t)$,},
        grid
      ]
      \foreach \i in {0, ..., 4} {
        \addplot[blue, no markers] table [x=t, y index=\i] {data/delta.txt};
      }
      \foreach \i in {0, 1} {
        \addplot[red, no markers, dashed] table [x=t, y index=\i] {data/hdelta.txt};
      }
    \end{axis}

    \begin{axis}[
        name=V,
        at=(delta.below south west),
        anchor=above north west,
        yscale=0.5,
        y tick label style={
          /pgf/number format/.cd,
          fixed,
          fixed zerofill,
          precision=3,
          /tikz/.cd
        },
        legend pos=south east,
        legend entries={$V_{\cG}(t)$,,,,,$\hV_{\hcG}(t)$,},
        grid
      ]
      \foreach \i in {0, ..., 4} {
        \addplot[blue, no markers] table [x=t, y index=\i] {data/V.txt};
      }
      \foreach \i in {0, 1} {
        \addplot[red, no markers, dashed] table [x=t, y index=\i] {data/hV.txt};
      }
    \end{axis}

    \begin{axis}[
        name=theta,
        at=(V.below south west),
        anchor=above north west,
        yscale=0.5,
        y tick label style={
          /pgf/number format/.cd,
          fixed,
          fixed zerofill,
          precision=2,
          /tikz/.cd
        },
        xlabel={$t$},
        legend pos=south east,
        legend entries={$\theta_{\cG}(t)$,,,,,$\htheta_{\hcG}(t)$,},
        grid
      ]
      \foreach \i in {0, ..., 4} {
        \addplot[blue, no markers] table [x=t, y index=\i] {data/theta.txt};
      }
      \foreach \i in {0, 1} {
        \addplot[red, no markers, dashed] table [x=t, y index=\i] {data/htheta.txt};
      }
    \end{axis}
  \end{tikzpicture}
  \caption{Initial value response of the original power system from
    \Cref{fig:power-system} and a reduced system obtained by aggregating with
    partition $\{\{1, 2, 3\}, \{4, 5\}\}$.
    Original system's parameters are $\chi_i = \chi_{ij} = 1$ for all $i, j$,
    $M = D = I_5$, $f = 0$, and $E = \one_5$.
    The initial value is $\delta(0) = (0, 0.1, 0.2, 0.3, 0.4)$ and
    $\dot{\delta}(0) = 0$.
  }\label{fig:simulation}
\end{figure}

\section{Conclusions}
We analyzed power systems consisting of generators and buses.
We introduced a notion of synchronization for a pair of generators and two for
a partition of the set of generators.
We proved equivalent conditions depending on the Kron-reduced system being
symmetrical or equitable.
This additionally gives a relation between symmetrical matrices and equitable
partitions.
We showed how a synchronized power systems can be exactly approximated with a
reduced system by aggregating generators and their buses.
Furthermore, this provides an aggregation-based reduction method for arbitrary
power systems, although finding bounds for the approximation error remains an
open problem.

\section*{Appendix}
\begin{lemma}\label{thm:X1=1}
  For $X$ as in~\eqref{eq:X-def}, we have $X \one = \one$.
\end{lemma}
\begin{proof}
  After some algebraic manipulation, it is clear $X \one = \one$ is equivalent
  to $\myparen{L_{11} - L_{12} L_{22}^{-1} L_{12}^{\T}} \one = 0$, which follows
  from $L \one = 0$.
\end{proof}
\begin{lemma}\label{thm:symmetrical-equivalence}
  Let $A \in \R^{n \times n}$ be a symmetric matrix and
  $i, j \in \{1, 2, \ldots, n\}$ such that $i \neq j$.
  Then $A \in \cS_{ij}$ if and only if $a_{ii} = a_{jj}$ and $a_{ik} = a_{jk}$
  for all $k \neq i, j$.
\end{lemma}
\begin{proof}
  From the definition, it can be seen that $A \in \cS_{ij}$ is equivalent to
  $a_{ii} = a_{jj}$, $a_{ij} = a_{ji}$, $a_{ik} = a_{jk}$, and $a_{ki} = a_{kj}$
  for all $k \neq i, j$.
  Using that $A$ is symmetric, the conditions of the Lemma follow.
\end{proof}
\begin{lemma}\label{thm:symmetrical-properties}
  Let $A, B \in \cS_{ij}$ for some $i, j \in \{1, 2, \ldots, n\}$ such that
  $i \neq j$ and $\alpha, \beta \in \R$.
  Then,
  \begin{enumerate}
  \item $\alpha A + \beta B \in \cS_{ij}$,
  \item $A B \in \cS_{ij}$, and
  \item if $A$ is nonsingular, then $A^{-1} \in \cS_{ij}$.
  \end{enumerate}
\end{lemma}
\begin{proof}
  Follows directly from the definition of $\cS_{ij}$ in~\eqref{eq:symmat}.
\end{proof}
\begin{lemma}\label{thm:Gamma-symmetrical}
  Let $i, j \in \{1, 2, \ldots, n\}$ be such that $i \neq j$.
  We have $\Gamma \in \cS_{ij}$ if and only if $\LD \in \cS_{ij}$ and
  $L_{11} - L_{12} L_{22}^{-1} L_{12}^{\T} \in \cS_{ij}$.
\end{lemma}
\begin{proof}
  \fbox{$\Leftarrow$} Follows from \Cref{thm:symmetrical-properties}.

  \fbox{$\Rightarrow$} First we show that $\LD \in \cS_{ij}$.
  Using $\Gamma = \LD X$, $\Pi_{ij} \one = \one$, and $X \one = \one$, from
  $\Gamma \Pi_{ij} \one = \Pi_{ij} \Gamma \one$ it follows that
  $\LD \one = \Pi_{ij} \LD \one$.
  Since $\LD$ is a diagonal matrix, from this we see that $\LD \in \cS_{ij}$.
  Now $L_{11} - L_{12} L_{22}^{-1} L_{12}^{\T} \in \cS_{ij}$ follows from
  \Cref{thm:symmetrical-properties}.
\end{proof}

\section*{Acknowledgment}
The work of the first author was supported by a research grant of the
``International Max Planck Research School (IMPRS) for Advanced Methods in
Process and System Engineering (Magdeburg)''.
The work of the second and sixth author was partially supported by JSPS
Grant-in-Aid for Scientific Research (A) 26249062 and JST CREST JPMJCR15K1.

\bibliographystyle{IEEEtran}
\bibliography{root_arxiv}
\end{document}